\documentclass[a4paper,USenglish,cleveref,autoref,thm-restate,final]{lipics-v2021}
\usepackage{mathtools}

\title{The Pseudo-Reachability Problem for Diagonalisable Linear
	Dynamical Systems} 

\author{Julian D'Costa}{Department of Computer Science, University of Oxford, UK}%
{julianrdcosta@gmail.com}{https://orcid.org/0000-0003-2610-5241}{emmy.network foundation under the aegis of the Fondation de Luxembourg.}

\author{Toghrul Karimov}%
{Max Planck Institute for Software Systems, Saarland Informatics Campus, Germany}%
{toghs@mpi-sws.org}{https://orcid.org/0000-0002-9405-2332}{}

\author{Rupak Majumdar}%
{Max Planck Institute for Software Systems, Kaiserslautern, Germany}%
{rupak@mpi-sws.org}{https://orcid.org/0000-0003-2136-0542}{DFG grant 389792660 as part of TRR 248 (see
	\url{https://perspicuous-computing.science}).}

\author{Jo\"el Ouaknine}%
{Max Planck Institute for Software Systems, Saarland Informatics Campus, Germany}%
{joel@mpi-sws.org}{https://orcid.org/0000-0003-0031-9356}{%
	DFG grant 389792660 as part of TRR 248 (see
	\url{https://perspicuous-computing.science}). 
	Jo\"el Ouaknine is also affiliated with Keble College, Oxford as \href{http://emmy.network/}{\texttt{emmy.network}} Fellow.}

\author{Mahmoud Salamati}%
{Max Planck Institute for Software Systems, Kaiserslautern, Germany}%
{msalamati@mpi-sws.org}{https://orcid.org/0000-0003-3790-3935}{}

\author{James Worrell}%
{Department of Computer Science, University of Oxford, UK}%
{jbw@cs.ox.ac.uk}{https://orcid.org/0000-0001-8151-2443}{}

\authorrunning{J. D'Costa et al.}
\Copyright{J. D'Costa, T. Karimov, R. Majumdar, J. Ouaknine, M. Salamati, S. Soudjani and J. Worrell} %TODO mandatory, please use full first names. LIPIcs license is "CC-BY";  http://creativecommons.org/licenses/by/3.0/
 %TODO mandatory; please add comma-separated list of keywords

\category{} %optional, e.g. invited paper

\relatedversion{https://arxiv.org/abs/2204.12253} %optional, e.g. full version hosted on arXiv, HAL, or other respository/website
\nolinenumbers %uncomment to disable line numbering

\EventEditors{Stefan Szeider, Robert Ganian, and Alexandra Silva}
\EventNoEds{3}
\EventLongTitle{47th International Symposium on Mathematical Foundations of Computer Science (MFCS 2022)}
\EventShortTitle{MFCS 2022}
\EventAcronym{MFCS}
\EventYear{2022}
\EventDate{August 22--26, 2022}
\EventLocation{Vienna, Austria}
\EventLogo{}
\SeriesVolume{241}
\ArticleNo{75}
\newtheorem{problem}[theorem]{Problem}

\def\reals{\mathbb{R}}
\def\nats{\mathbb{N}}
\def\rats{\mathbb{Q}}
\def\algebraics{{\overline{\mathbb{Q}}}}
\def\complex{\mathbb{C}}

\def\integers{\mathbb{Z}}

\def\restr{{\upharpoonright}}

\def\theory#1{{\mathfrak{Th}(#1)}}

\def\realexp{{\mathbb{R}_{\exp}}}

\def\set#1{{\{ #1 \}}}

\def\O{{\mathcal{O}}}
\def\PO{{\widetilde{\mathcal{O}}}}

\def\2dvec#1#2{\begin{bmatrix}
		#1\\
		#2
\end{bmatrix}}
\def\zerovec{{\mathbf{0}}}

\newcommand{\torus}{\mathbb{T}}
\newcommand{\aepsilon}{\mathcal{A}_\epsilon}

\newcommand{\B}{\mathcal{B}}
\newcommand{\T}{\mathcal{T}}

\def\theory{{\mathbb{R}_{\mathsf{MW}}}}

\begin{document}
	
	\maketitle
	
	%TODO mandatory: add short abstract of the document
	\begin{abstract}
		We study fundamental reachability problems on pseudo-orbits of linear dynamical systems. 
		Pseudo-orbits can be viewed as a model of computation with limited precision and pseudo-reachability can be thought of as a robust version of classical reachability. 
		Using an approach based on $o$-minimality \\of $\reals_{\exp}$ we prove decidability of the discrete-time pseudo-reachability problem with arbitrary semialgebraic targets for diagonalisable linear dynamical systems.
		We also show that our method can be used to reduce the continuous-time pseudo-reachability problem to the (classical) time-bounded reachability problem, which is known to be conditionally decidable.
		%Using an approach based on $o$-minimality of $\reals_{\exp}$, we show decidability of the pseudo-reachability problem with semialgebraic targets for both discrete-time and continuous-time diagonalisable linear dynamical systems.
	\end{abstract}

	%------------------macros
	\newcommand{\RM}[1]{\textcolor{red}{\textbf{RM:} #1}}
	\newcommand{\MS}[1]{\textcolor{blue}{\textbf{MS:} #1}}
	\newcommand{\JD}[1]{\textcolor{purple}{\textbf{JD:} #1}}
	\newcommand{\TK}[1]{\textcolor{orange}{\textbf{TK:} #1}}
	\newcommand{\SES}[1]{\textcolor{magenta}{\textbf{SS:} #1}}
	
	\section{Introduction}

A \emph{discrete-time linear dynamical system} (LDS) is given by an update matrix $M \in \rats^{d\times d}$ and a starting point $s \in \rats^d$. An LDS describes a system whose state contains $d$ rational numbers and evolves linearly. The \emph{orbit} of such a system is the infinite sequence $\langle s, Ms, M^2s, \ldots \rangle$ of points in $\rats^d$. Orbits of LDS arise in many areas of computer science and mathematics, including verification of linear loops \cite{POPL22}, automata theory \cite{barloy_et_al:LIPIcs:2020:11652}, and the theory of linear recurrence sequences \cite{joelReview}.

A fundamental problem about LDS is the question of deciding, given a system $\langle M, s\rangle$ and a semialgebraic target set $S \subseteq \reals^d$, whether there exists $n$ such that $M^ns \in S$.
This problem is known as the \emph{reachability problem for LDS} and has been studied extensively over the last few decades.
In their seminal work, Kannan and Lipton showed that the point-to-point reachability problem, i.e., the case in which $S$ is a singleton, is decidable in polynomial time. 
At the same time they observed that the case in which $S$ is a $(d-1)$-dimensional subspace of $\reals^d$ (i.e.~a~hyperplane) is equivalent to the famous \emph{Skolem problem} whose decidability remains open to this day. 
The Skolem problem asks, given a linear recurrence sequence defined by a recurrence relation $u_{n+d} = a_1u_n + \ldots + a_du_{n+d-1}$ and initial values $u_0, \ldots, u_{d-1}$, to decide whether there exists~$n$ such that $u_n = 0$.
In addition to this Skolem-hardness, the difficulty of settling the reachability problem was further demonstrated by the results of \cite{Ouaknine2014}, which show that solving the reachability problem with halfspace targets, known as the \emph{positivity problem}, would entail major mathematical breakthroughs in the field of Diophantine approximation.

The reachability problem is defined with reference to the exact dynamics of an LDS\@. Since computational systems typically operate with finite precision, it is natural to consider an alternate notion of reachability involving so-called \emph{pseudo-orbits}.
The notion of pseudo-orbit is an important conceptual tool in dynamical systems that was 
introduced
by Anosov~\cite{Anosov}, Bowen~\cite{Bowen1}, and
Conley~\cite{Conley78}, and was used by the latter to prove what is
sometimes called the fundamental theorem of dynamical systems. Given an LDS $\langle M, s \rangle$, a sequence $\langle x_n \mid n \in \nats \rangle$ is an \emph{$\epsilon$-pseudo-orbit} of $s$ under $M$ if $x_0 = s$ and $\|Mx_n - x_{n+1}\| < \epsilon$ for all $n \in \mathbb{N}$. In other words, in a pseudo-orbit one considers an enlarged transition relation that is obtained by considering the dynamical system up to precision $\epsilon$. Given $\epsilon>0$, a set $S$ is said to be \emph{$\epsilon$-pseudo-reachable} if there exists an $\epsilon$-pseudo-orbit $\langle x_0 = s, x_1, x_2, \ldots \rangle$ of $s$ under $M$ that reaches $S$. We further say that $S$ is \emph{pseudo-reachable} if $S$ is $\epsilon$-pseudo-reachable for every $\epsilon>0$. If a set $S$ of error states is not pseudo-reachable then we can consider the system as being safe if implemented with sufficient precision, while if $S$ is pseudo-reachable, it means that no finite amount of precision suffices to make the system reliably safe. 

Recently, D'Costa et al.~\cite{dcosta_et_al:LIPIcs.MFCS.2021.34} considered the pseudo-reachability problem and, somewhat surprisingly, showed decidability in cases where $S$ is a point (the pseudo-orbit problem), a hyperplane (the pseudo-Skolem problem) or a halfspace (the pseudo-positivity problem). Their proof of the first result relies on an exact characterisation of $\epsilon$-pseudo-orbits. Their solution to the latter two problems, however, depends heavily on the fact that a hyperplane (a halfspace) can be defined using a single equality (inequality), an approach which unfortunately cannot be generalised to arbitrary semialgebraic targets.
\textbf{In this work, we develop a novel logical approach to show the decidability of the pseudo-reachability problem for diagonalisable systems with arbitrary semialgebraic targets.} 
%We also show that the $\epsilon$-reachability problem for diagonalisable systems is hard with respect to an open subclass of the Skolem problem.

\subsection{High-level proof sketch of our approach}

Our solution to the diagonalisable pseudo-reachability problem can be summarised as follows.
Let $\PO_{\epsilon}(n)$ denote the set of all points that are reachable exactly at time $n$ via an $\epsilon$-pseudo-orbit. 
The pseudo-reachability problem then consists in checking whether the sentence $\Phi \coloneqq \forall \epsilon. \, \exists n\in \nats: \PO_{\epsilon}(n) \cap S \neq \emptyset$ is true.
In this form, $\Phi$ is not amenable to application of logical methods as it involves both integer and real-valued variables, in addition to exponentiation with a complex base (coming from non-real eigenvalues of $M$). 
We therefore first move to the continuous domain and construct an abstraction $\aepsilon(t)$ for $t \in \reals^{\geq 0}$ that is definable in $\realexp$ such that $\aepsilon(n) \supseteq \PO_{\epsilon}(n)$ for all $n \in \nats$.
We then investigate the values of $\epsilon$ and $t$ that make $\Psi(\epsilon, t) \coloneqq \aepsilon(t)\cap S \neq \emptyset$ true.
We show that by the $o$-minimality of $\realexp$, either for every $\epsilon>0$ there exists $T$ such that for all $t > T$, $\Psi(\epsilon, t)$ holds, or the pseudo-reachability problem is equivalent to a finite-horizon reachability problem that is easily solvable.
In the former case, it follows that for every $\epsilon > 0$, $\Psi(\epsilon, n)$ holds for all sufficiently large integer values $n$, thus establishing a bridge back to the discrete setting.
We conclude by showing that in this case, $S$ is pseudo-reachable.
Intuitively, the idea is to use the universal quantification over $\epsilon$ to argue that if $S$ can be reached using an $\epsilon/2$-abstraction at all but finitely many time steps, then it can be reached by an $\epsilon$-pseudo-orbit, in fact at infinitely many possible time steps. 
The importance of the universal quantification is also illustrated by the following hardness result.
For any fixed $\epsilon>0$, it is decidable whether $\exists n \in \nats: \, \aepsilon(n) \cap S \neq \emptyset$ holds, whereas the $\epsilon$-pseudo-reachability problem of determining whether $\exists n \in \nats: \, \PO_{\epsilon}(n) \cap S \neq \emptyset$ holds is hard with respect to (a hard subclass of) the Skolem problem, as shown in \autoref{hardness}. 

The approach outlined above can be adapted to solve a few other related problems about linear dynamical systems.
An example would be the \emph{robust reachability problem} recently considered by Akshay et al. in \cite{akshay_et_al:LIPIcs.STACS.2022.5}: given an LDS $\langle M, s\rangle$ and a semialgebraic target $S$, decide whether for all $\epsilon>0$ there exists a point $s'$ in the $\epsilon$-neighbourhood of $s$ whose orbit reaches $S$.
This problem can be thought of as a modification of the pseudo-reachability problem where only one perturbation is allowed at the very beginning. 
Due to this simplification, we are able to show, in \autoref{sec::discrete_robust_reach}, full decidability (that is, without the restriction to diagonalisable systems) of the robust reachability problem.
Finally, because the first step of our solution is to translate the problem into the continuous domain, the continuous versions of both the pseudo-reachability problem (discussed in \autoref{sec:cont_pseudo_reach}) and the robust reachability problem (discussed in \autoref{sec::cont_robust_reach}) can be handled using the same approach, arguably more naturally.
For the former, because we proceed by reducing the pseudo-reachability problem to bounded-time reachability problem, the decidability result assumes Schanuel's conjecture.

%Research on the reachability problem then proceeded by considering dimension restrictions either on the system $\langle M, s\rangle$ or on the target set $S$. Chonev et al. confirmed Kannan and Lipton's conjecture by showing decidability of the case where $S$ is a polytopic subset of $\reals^d$ contained in a subspace of dimension at most $3$. 
%Almagor et al., on the other hand, showed the full decidability of the reachability problem for LDS of dimension $d \leq 3$. 
%Baier~et~al. isolated semialgebraic sets of (intrinsic) dimension 1 as another family of target shapes for which reachability is decidable in all ambient dimensions. 
%In both cases the results are tight in the sense that increasing the dimension restriction by one makes the reachability problem Skolem-hard. This line of research culminated with the result of Karimov et al. that the \emph{full $\omega$-regular model-checking problem} is decidable for specifications over target sets that belong to a class for which reachability problem is known to be decidable.
	
	\section{Mathematical tools}

We write $B(c, r)$ for the closed $\ell_2$-ball of radius $r$ centred around $c \in \reals^d$ and $\zerovec \in \reals^d$ for the $d$-dimensional zero vector. We denote by $\torus \subseteq \complex$ the unit circle in the complex plane and by~$||x||$ the $\ell_2$-norm of a vector $x\in \reals^d$.

\subsection{First-order logic}

We denote by $\reals_0$ the (structure of) real numbers with addition and multiplication, by $\realexp$ the real numbers with addition, multiplication and (unbounded) exponentiation and by $\reals_{\exp, \cos \restr [0, T]}$ the real numbers with exponentiation and bounded (in input, by some~$T > 0$) trigonometric functions. By the Tarski-Seidenberg theorem, the theory of $\reals_0$ admits effective quantifier elimination and is therefore decidable. The theories of $\realexp$ and $\reals_{\exp, \cos \restr [0, T]}$ are known to be decidable subject to Schanuel's conjecture (see, e.g., \cite{Lang1966}) in transcendental number theory \cite{Macintyre1996, Wilkie1997}. However $\reals_{\exp, \cos \restr [0, T]}$ (and hence $\realexp$ and $\reals_0$) are unconditionally known to be $o$-minimal \cite{vandenDries}. That is, any subset of $\reals$ definable using arithmetic operations, real exponentiation and bounded trigonometric functions is a finite union of intervals. In~particular, any subset of $\reals^{\geq0}$ definable in this way is either bounded or contains all sufficiently large real numbers. For the discrete-time problems considered in this paper we will only need to work with $\realexp$. We will need $\reals_{\exp, \cos \restr [0, T]}$ only when considering the classical bounded-time reachability problem for continuous-time linear dynamical systems.

A \emph{semialgebraic} set is a subset of $\reals^d$ definable (without parameters) in $\reals_0$. We say that a function $\varphi: \reals^{l}\to\reals^{m}$ is semialgebraic if its graph is a semialgebraic subset of $\reals^{l+m}$. Intuitively, semialgebraic functions are exactly the functions that can be specified using arithmetic and logical operations over the real numbers.

\begin{comment}
	\JD{copied from continuous section}
	Our decidability results for the continuous-time setting will require Schanuel's Conjecture for the complex numbers,
	a unifying conjecture in transcendental number theory
	(see, e.g., \cite{Lang1966}).
	Recall that a \emph{transcendence basis} of a field extension $L/K$ is a subset 
	$S \subseteq L$ such that $S$ is algebraically independent over $K$ and $L$ is algebraic over $K(S)$. 
	The \emph{transcendence degree} of $L/K$ is the unique cardinality of some basis.
	
	\begin{conjecture}[\textbf{SC}]
		Let $a_1,\ldots, a_n$ be complex numbers that are linearly independent over rational numbers $\rats$. 
		Then the field $\rats(a_1,\ldots,a_n,e^{a_1},\ldots,e^{a_n})$ has transcendence degree at least $n$ over $\rats$.
	\end{conjecture}
	
	%An important consequence of Schanuel's conjecture is that the theory of reals $(\reals, 0, 1, +, \cdot, \leq)$
	%remains decidable when extended with the exponential and trigonometric functions over bounded domains.
	One major consequence of Schanuel's conjecture is that the theory of reals with the real exponentials and trigonometric functions restricted to a bounded interval is decidable.
	
	\begin{theorem}[Macintyre and Wilkie \cite{Macintyre1996,Wilkie1997}]
		\label{th:macintyre}
		Assume $\textbf{SC}$.
		For any $n\in \nats$, the theory $\theory :=(\reals, \exp\upharpoonright [0,n], \sin\upharpoonright [0,n], \cos\upharpoonright [0,n])$ is decidable.
	\end{theorem}
\end{comment}

\subsection{Kronecker's theorem and its applications}
\label{sec::Kronecker}

The analysis of problems about linear dynamical systems often reduces to that of the orbit $\langle \Gamma^n \mid n \in \nats\rangle$ where $\Gamma^n = (\gamma_1^n, \ldots, \gamma_k^n)$ for $\gamma_1, \ldots, \gamma_k \in \torus$.
Let $\T = \operatorname{cl}(\{\Gamma^n : n \in \nats\})$ be the topological closure of this discrete orbit.
The set $\T$ is semialgebraic and well-understood with the help of Kronecker's theorem in simultaneous Diophantine approximation \cite{Hardy:1999}.
\begin{theorem}[Kronecker]
	\label{thm::Kronecker}
	Let $\theta_1, \ldots, \theta_k, \varphi_1, \ldots, \varphi_k \in \reals$ be such that for any  $a_1, \ldots, a_k \in \integers$,
	\[
	\sum_{i=1}^k a_i \theta_i \in \integers \Rightarrow \sum_{i=1}^k a_i \varphi_i \in \integers.
	\]
	For any $\epsilon>0$ there exist infinitely many $n \in \nats$ such that $\{n\theta_i-\varphi_i\} < \epsilon$ for all $1\leq i \leq k$, where $\{x\}$ denotes the distance from $x \in \reals$ to the nearest integer.
\end{theorem} 
To apply this theorem to our situation, let 
\[
\T = \{(z_1, \ldots, z_k): \forall a_1,\ldots,a_k \in \integers: \, \gamma_1^{a_1}\cdots\gamma_k^{a_k} = 1 \Rightarrow  z_1^{a_1}\cdots z_k^{a_k} = 1 \}.
\]
For $z =(z_1, \ldots, z_k) \in \T$, by considering $\theta_i = \frac{\operatorname{arg}(\gamma_i)}{2\pi}$ and $\varphi_i = \frac{\operatorname{arg}(z_i)}{2\pi}$ for $1 \leq  i \leq k$ we can deduce that for each $\epsilon > 0$ there exists $n$ such that $||z - \Gamma^n||<\epsilon$ and hence the orbit $\langle \Gamma^n \mid n \in \nats\rangle$ is dense in $\T$.
On the other hand, using Masser's deep results \cite{Mas88} about multiplicative relations between algebraic numbers one can compute, in polynomial time, a finite basis for $\{(a_1, \ldots, a_k) \in \integers^k : \gamma_1^{a_1}\cdots\gamma_k^{a_k}=1\}$.
Hence $\T$ is closed, semialgebraic and effectively computable. 
It then follows that $\T = \operatorname{cl}(\{\Gamma^n : n \in \nats\})$.

We will also need the following lemma which is a consequence of the effective computability of $\T = \operatorname{cl}(\{\Gamma^n : n \in \nats\})$ as a semialgebraic set.
\begin{lemma}
	Let $R = \operatorname{diag}(\Lambda_1,\ldots, \Lambda_k) \in \reals^{2k\times 2k}$ be a block diagonal matrix where $\Lambda_i$ is an algebraic rotation matrix for $1 \leq i \leq k$. The closure of the set $\{R^n x : n \in \nats\}$, for $x$ with algebraic entries, is semialgebraic and effectively computable.
\end{lemma}
The proof follows immediately from diagonalising $R^n$ and observing that all eigenvalues of~$R$ are algebraic numbers in $\torus$.
	
	%\input{discrete-setup}
	
	% !TEX root = main.tex
\section{Decidability for discrete-time diagonalisable systems}	
\label{sec::disc-decidability}
In this section we prove our main result: the decidability of the pseudo-reachability problem for discrete-time diagonalisable \emph{affine dynamical systems}, which are a generalisation of~LDS\@.\\ 
The reason we consider affine systems is that the well-known homogenisation trick (increasing the dimension by one and adding a coordinate that is always equal to $1$) used for reducing the classical reachability problem for affine systems to the reachability problem for LDS doesn't work for the pseudo-reachability problem: when perturbations are allowed, one cannot force a coordinate to remain constant.
Hence affine systems require separate treatment. 

\begin{problem}[pseudo-reachability]
	Let $M \in \rats^{L \times L}$ be an update matrix, $s \in \rats^L$ be a starting point, $b \in \rats^L$ be an affine term and $S \subseteq \reals^L$ be a semialgebraic target set. A sequence $\langle x_0 = s, x_1, x_2 \ldots \rangle$ is an \emph{$\epsilon$}-pseudo-orbit of $s$ if $||Mx_n + b - x_{n+1}|| \leq \epsilon$ for all $n$. 
	The pseudo-reachability problem asks: given $M, b, s$ and $S$, decide whether for each $\epsilon>0$ there exists an $\epsilon$-pseudo-orbit of $s$ 
	%\textcolor{red}{RM: the next phrase is superfluous. pseudo orbit already defines the dynamics, and the below equality does not hold. Omit.}
	%under the dynamics $x_{n+1}=Mx_n+b$ 
	%\textcolor{red}{RM: end}
	that reaches the set $S$. 
\end{problem}
Let $\PO_{\epsilon}(n)$ denote the set of all points that are reachable via an $\epsilon$-pseudo-orbit of $s$ under the map $x \mapsto Mx +b$ at time $n$. 
Since $\PO_{\epsilon}(0) = \{s\}$ and $\PO_{\epsilon}(n+1) = M\PO_{\epsilon}(n) + b + \epsilon B(\zerovec, 1)$, by~induction we can show that $\PO_{\epsilon}(n) =  M^ns + \sum_{i=0}^{n-1}M^ib + \epsilon \sum_{i=0}^{n-1}M^iB(\zerovec, 1)$. 
The pseudo-reachability problem is then equivalent to determining the truth of $\forall \epsilon. \, \exists n : \PO_{\epsilon}(n) \cap S \neq \emptyset$.
Here $B(\zerovec, 1)$ can be viewed as a set of ``control inputs'', and the pseudo-reachability problem can be viewed as the problem of determining whether $S$ can be reached using arbitrarily small control inputs.
The next lemma shows that we can in fact, choose any reasonable  control set.

\begin{lemma}[Invariance under change of the control set]
	\label{convenient-control-set}
	Let $\B \subseteq \reals^L$ be a bounded set containing an open ball around the origin.
	\begin{enumerate}
		\item The pseudo-reachability problem as defined above is equivalent to the problem of determining whether
		\[
		\forall \epsilon. \, \exists n :( M^ns + \sum_{i=0}^{n-1}M^ib + \epsilon \sum_{i=0}^{n-1}M^i\B) \cap S \neq \emptyset.
		\]
		\item  We may assume the matrix $M$ is in real Jordan form.
	\end{enumerate}
\end{lemma}
\begin{proof}
	Since $\B$ is assumed to be bounded and to contain an open neighbourhood around the origin, there must exist constants $C_1, C_2$ such that $C_1B(\zerovec, 1) \subseteq \B \subseteq C_2B(\zerovec, 1)$.
	Hence
	\[
	C_1\epsilon \sum_{i=0}^{n-1}M^iB(\zerovec, 1) \subseteq \epsilon \sum_{i=0}^{n-1}M^i\B \subseteq C_2\epsilon \sum_{i=0}^{n-1}M^iB(\zerovec, 1).
	\]
	The proof of (1) then follows from the fact that $\epsilon$ is universally quantified: one can simulate (i) an $\epsilon$-pseudo-orbit with control set $\B$ using a $C_2\epsilon$-pseudo-orbit with control set $B(\zerovec, 1)$ and (ii) an $\epsilon$-pseudo-orbit with control set $B(\zerovec, 1)$ using a $\epsilon/C_1$-pseudo-orbit with control set $\B$. 
	Proof of (2) follows from observing that multiplying $B(\zerovec, 1)$ by an invertible change of basis matrix results in a bounded control set containing a neighbourhood around $\zerovec$.
\end{proof}
Observe that the change of the control set described above is not applicable when $\epsilon$ is fixed, as in the $\epsilon$-pseudo-reachability problem discussed in \autoref{hardness}.

\subsection{A closed form for $\PO_{\epsilon}(n)$}
\label{sec::closed-form}
We now use \autoref{convenient-control-set} to choose a control set that results in  $\PO_{\epsilon}(n)$ with a convenient first-order closed form: observe that the na\"ive formulation above involves the term $\sum_{i=0}^{n-1}M^iB(\zerovec, 1)$ which is not ``first-order''.

Assume $M$ is diagonalisable and in real Jordan form: $M = \operatorname{diag}(\Lambda_1, \ldots, \Lambda_k, \rho_{k+1}, \ldots, \rho_d)$.
That is, $M$ consists of $d$ block, the first $k$ of which have dimension $2\times2$ and a pair of non-real conjugate eigenvalues, whereas the remaining blocks are $1 \times 1$ and real.
Write $\rho_j$ for the spectral radius of the $j$th block.
We can factor $M$ into a ``scaling'' and a ``rotation'' as $M = D R$ where $D = \operatorname{diag}(\rho_1, \rho_1, \ldots, \rho_k, \rho_k, \rho_{k+1}, \rho_{k+2}, \ldots, \rho_{d})$ is diagonal and $R$ is a block-diagonal matrix that consists of blocks that are either $2\times2$ rotation matrices or $1\times 1$ and equal to $\begin{bmatrix}
	\pm 1
\end{bmatrix}$.
Hereafter we will be using the convenient ``rotation-invariant'' control set
\[
\B = \prod_{j=1}^k B((0,0),1) \times \prod_{j=k+1}^{d}[-1,1] = \prod_{j=1}^d B(\zerovec, 1)
\]where $B((0,0),1)$ is the unit disc.
Observe that $\B$ is a product of $\ell_2$-balls that matches the block structure of $M$.
It follows that $R\B = \B$ and hence
\[
\PO_\epsilon(n)= D^n R^n s + \sum_{i=0}^{n-1}M^ib + \epsilon\sum_{i=0}^{n-1} D^iR^i\B = D^n R^n s + \sum_{i=0}^{n-1}M^ib + \epsilon \B(n)
\]
where $\B(n) = \sum_{i=0}^{n-1} D^i\B$. We then have
\[
\B(n)  = \sum_{i=0}^{n-1} D^i  \prod_{j=1}^d B(\zerovec, 1) = \sum_{i=0}^{n-1} \prod_{j=1}^d B(\zerovec, \rho_j^i) = \prod_{j=0}^d B(\zerovec, \sum_{i=0}^{n-1} \rho_j^i).
\]
Geometrically, the idea is that a $2\times2$ or a $1\times1$ block of $D$ maps an origin-centred disc (which corresponds to a symmetric interval in 1D) to an origin-centred disc, and a set-sum of such discs is again an origin-centred disc.  
Note that our ability to reason in this way crucially depends on the fact that $M$ is diagonalisable.
Finally, since $\sum_{i=0}^{n-1} \rho_j^i$ is either $\frac{\rho_j^n-1}{\rho_j-1}$ or $n \rho_j$, we can write
$\B(n) = \{z : \varphi(z,n, \rho_1^n, \ldots, \rho_d^n)\}$, where $\varphi$ is a semialgebraic predicate.

We can apply the blockwise summation technique, distinguishing between the cases where the spectral radius of the block is 1 or different from 1, to the term $\sum_{i=0}^{n-1}M^ib$ to obtain the closed form $\sum_{i=0}^{n-1}M^ib  = D^nR^nx' + cn + d$, where $x'$, $c$ and $d$ only depend on $M$ and $b$. 
We then fold $s$ and $x'$ into a new, fictive starting point $x$ to obtain the final closed form 
\[
\PO_\epsilon(n)= D^n R^n x + cn +d + \epsilon \B(n).
\]
In order to solve the pseudo-reachability problem, we henceforth consider the problem of determining the truth of the sentence $\forall \epsilon > 0. \, \exists n: (D^n R^n x + cn +d + \epsilon \B(n)) \cap S \neq 0$, where all the input vectors and matrices have real algebraic entries.

%Let $M = \operatorname{diag}(\Lambda_1, \ldots, \Lambda_k, \rho_{k+1}, \ldots, \rho_{d})$ be a diagonalisable matrix in real Jordan form, $s$~be a starting point, $b$ be an affine term and $S$ be a semialgebraic target.
%As discussed earlier,
%\[
%\PO_\epsilon(n) = M^ns + \sum_{i=0}^{n-1}M^ib + \epsilon \sum_{i=0}^{n-1}M^i\B = M^nx+cn+d+\epsilon\B(n)
%\]
%is the set of all points reachable after exactly $n$ steps via an $\epsilon$-pseudo-orbit.
%%(Here $B((0,0),1)$ is the unit disc 2-dimensional ball in the $\ell_2$ norm. We choose this particular ``product of $\ell_2$ balls'' control set because it is conveniently invariant under rotation.)
%%Recall that $\PO_\epsilon(n) = M^ns + \sum_{i=0}^{n-1}M^ib + \epsilon \sum_{i=0}^{n-1}M^i\B$ is the set of all points reachable after $n$ steps.
%%Let $x, c, d$ be such that $\PO_\epsilon(n)=M^nx + cn + d + \epsilon \sum_{i=0}^{n-1}M^i\B$.
%We show how to decide whether $\forall \epsilon>0.\, \exists n: \PO_\epsilon(n) \cap S \neq \emptyset$.
%It follows that the pseudo-reachability problem for discrete-time diagonalisable affine systems is decidable.
%
%Write $M=DR$ where $D$ is diagonal and $R$ is a block-diagonal matrix whose diagonal entries are either $1\times 1$ or $2\times 2$ rotation matrices. 

\subsection{Passing to the abstraction}
The expression for $\PO_{\epsilon}(n)$ contains the term $D^nR^nx$, which is the last obstacle to obtaining an expression which we can attack using known results about theories of real numbers. 
To address this issue we resort to abstracting $\PO_{\epsilon}(n)$. 
Let
\[
\mathcal{T} := \operatorname{cl} \left( \{R^nx : n \in \nats\} \right)
\textrm{ and }
\aepsilon(n) := D^n\mathcal{T} + cn +d + \epsilon \B(n)
\]
where $\mathcal{T}$ is the closure of the orbit of $x$ under $R$, and is semialgebraic and effectively computable by the discussion in \autoref{sec::Kronecker}.
Moreover, recall that by Kronecker's theorem for every $z \in \mathcal{T}$ and $\epsilon>0$ there exist infinitely many integers $0<n_1 <n_2< \ldots$ such that $||R^{n_i}x-z||<\epsilon$ for all $i$.

Here $\aepsilon(n)$ acts as an abstraction of $\PO_\epsilon(n)$. In particular, for all $\epsilon>0$ and $n\in \nats$ we have $\aepsilon(n) \supseteq \PO_\epsilon(n)$. 
Observe that $\aepsilon(n) = \{z : \varphi(z, \epsilon, n, \rho_1^n, \ldots, \rho_d^n)\}$ for a semialgebraic predicate $\varphi$.
Viewing $\aepsilon(n)$ as a proxy for $\PO_\epsilon(n)$, we arrive at the following dichotomy.
\begin{lemma}
	\label{lem::discrete-dichotomy}
	Either
	\begin{enumerate}
		\item for every $\epsilon > 0$ there exists $N_\epsilon$ such that for all $n > N_\epsilon$, $\aepsilon(n)$ intersects $S$, or
		\item there exist $N$ and $\epsilon>0$, both effectively computable, such that $\aepsilon(n)$ does not intersect $S$ for all $n > N$. 
	\end{enumerate}
	Moreover, it can be effectively determined which case holds.
\end{lemma}
\begin{proof} 
	First we show that the dichotomy holds, putting the issues of effectiveness aside. 
	Let 
	\[
	\Phi(\epsilon, n) = \bigvee_{\alpha \in A} \bigwedge_{\beta \in B} p_{\alpha, \beta}(\epsilon,n, \rho_1^n, \ldots, \rho_d^n) \bowtie_{\alpha, \beta} 0
	\]
	be a quantifier-free formula equivalent to $\aepsilon(n) \cap S \ne \emptyset$. 
	Such $\Phi(\epsilon, n)$ must exist because $\aepsilon(n)$ is semialgebraic with parameters from $\{\epsilon, n, \rho_1^n, \ldots, \rho_d^n\}$ and by the Tarski-Seidenberg theorem, each such set can be described using a quantifier-free formula of the form given above.
	Suppose Case~1 does not hold. Then there exists a particular $\epsilon>0$ such that $\Phi(\epsilon, n) $ does not hold for arbitrarily large~$n$. Treating $n$ as a continuous parameter, consider the set $\{n\in\reals_{\geq0}: \Phi(\epsilon,n) \textrm{ does not hold}\}$. By $o$-minimality of $\realexp$ this set is a finite union of intervals and and by the assumption that Case~1 does not hold,  it contains arbitrarily large integers. Hence it must contain all integers in $(N, \infty)$  for some $N \in \nats$. 
	That is, for all $n>N$ the formula $\Phi(\epsilon, n)$ does not hold.
	
	\proofsubparagraph{Effectiveness.}We now address the issues of effectiveness. Consider the formula 
	\[
	\Psi(\epsilon) = \exists N_\epsilon. \, \forall n > N_\epsilon: \Phi(\epsilon, n).
	\]
	We show that $\Psi(\epsilon)$ is equivalent to a formula $\psi(\epsilon)$ in the language of $\reals_0$. To determine which case holds it then remains to determine the truth value of the sentence $\forall \epsilon: \psi(\epsilon)$.
	
	By the $o$-minimality argument above, given $\epsilon>0$, each $p_{\alpha, \beta}(\epsilon,n, \rho_1^n, \ldots, \rho_d^n) \bowtie_{\alpha, \beta} 0$ either holds for finitely many integer values of $n$ or holds for all sufficiently large integer values $n$. 
	By elementary considerations it follows that $\Psi(\epsilon)$ is equivalent to 
	\[
	\bigvee_{\alpha \in A} \bigwedge_{\beta \in B} \, \exists N_\epsilon. \, \forall n > N_\epsilon:  p_{\alpha, \beta}(\epsilon,n, \rho_1^n, \ldots, \rho_d^n) \bowtie_{\alpha, \beta} 0.
	\]
	Hence it suffices to show how to construct a formula $\psi(\epsilon)$ in the language of $\reals_0$  that is equivalent to $ \exists N_\epsilon. \, \forall n > N_\epsilon:  p_{\alpha, \beta}(\epsilon,n, \rho_1^n, \ldots, \rho_d^n) \bowtie_{\alpha, \beta} 0$.
	For each $\epsilon>0$, the formula first tests if $p_{\alpha, \beta}(\epsilon)$ (as a polynomial in $d+1$ remaining variables) is identically zero.
	If yes, then $\varphi(\epsilon)$ is $\textrm{true}$ or $\textrm{false}$ depending only on $\bowtie_{\alpha, \beta}$.
	Otherwise, write $p_{\alpha, \beta}(\epsilon, n, \rho_1^n, \ldots, \rho_d^n) = \sum_{i=1}^k q_i(\epsilon, n)R_i^n$ where $q_i(\epsilon)$ is not identically zero for all $i$ and $R_1 > \cdots > R_k > 0$ are real algebraic numbers of the form $\rho_1^{p_1}\cdots\rho_d^{p_d}$ for $p_1, \ldots, p_d \in \nats$. 
	Since $\left|q_1(\epsilon, n)R_1^n\right|> \left| \sum_{i=2}^k q_i(\epsilon, n)R_i^n \right|$ for sufficiently large $n$, whether $p_{\alpha, \beta}(\epsilon,n, \rho_1^n, \ldots, \rho_d^n) \bowtie_{\alpha, \beta} 0$ holds for sufficiently large $n$ depends only on $q_1(\epsilon, n)$. 
	Hence we can choose $\psi(\epsilon)$ to be $\lim_{n\to \infty} q_1(\epsilon) \bowtie_{\alpha, \beta} 0$, which amounts to a sign condition on the coefficients of $q_1(\epsilon, n)$. 
	
	\proofsubparagraph{Computing $N$.} Finally, we show that in Case~2, the value $N$ can be effectively computed. To this end, by repeatedly trying smaller and smaller values of $\epsilon$ first compute a rational $e>0$ such that $\Psi(e)$ (equivalently, $\psi(e)$) does not hold. 
	To be able to compute $N$ it then suffices to compute, for a particular $(\alpha, \beta)$, a value $N_{\alpha,\beta}$ such that $p_{\alpha, \beta}(e,n, \rho_1^n, \ldots, \rho_d^n) \bowtie_{\alpha, \beta} 0$ does not hold for all $n > N_{\alpha,\beta}$, assuming that it does not hold for sufficiently large~$n$. We can then take $N$ to be the maximum of $N_{\alpha,\beta}$ over $(\alpha,\beta) \in A\times B$.
	
	To compute $N_{\alpha, \beta}$, consider $p \coloneqq p_{\alpha, \beta}(e)$. Assuming it is not identically zero (otherwise we can choose $N_{\alpha, \beta}$ to be any positive integer), write  $p(n, \rho_1^n, \ldots, \rho_d^n) = \sum_{i=1}^k q_i(n)R_i^n$ where $q_i$ is not identically zero for all $i$ and $R_1 > \cdots > R_k > 0$ are real algebraic. 
	Since $p_{\alpha, \beta}(e, n, \rho_1^n, \ldots, \rho_d^n) \bowtie_{\alpha, \beta} 0$ and hence $p(n, \rho_1^n, \ldots, \rho_d^n) \bowtie_{\alpha, \beta} 0$ do not hold for sufficiently large $n$, it must be the case that $q_1(n) \bowtie_{\alpha, \beta} 0$ does not hold for sufficiently large $n$. 
	Hence it remains to choose $N_{\alpha, \beta}$ large enough so that for all $n > N_{\alpha, \beta}$, $|q_1(n)R_1^n|$ dominates $\left| \sum_{i=2}^k q_i(n)R_i^n \right|$.
\end{proof}

\subsection{From the abstraction back to $\epsilon$-pseudo-orbits}
In this section we consider the relationship between the two cases of \autoref{lem::discrete-dichotomy} and our original pseudo-reachability problem.  \begin{comment}
	Case~2 is simpler: $S$ is pseudo-reachable if and only if
\end{comment} 
We start with Case~2. Observe that $\O_{\epsilon}(n)\subseteq \aepsilon(n)$ for every $n\in\nats$ and $\epsilon>0$. Therefore, when Case~2 holds, for any $n>N$ and $\epsilon>0$ the target set cannot be reached by $\O_{\epsilon}$. It remains to check pseudo-reachability at time steps $n\leq N$. 
We claim that $S$ is pseudo-reachable if and only
\[
\forall \epsilon. \, \exists n \leq N : (M^nx + cn +d + \epsilon\B(n) )\cap S \neq \emptyset.
\]
Let $\overline{S}$ denote the topological closure of $S$.
We show that the statement above is equivalent to $\exists n \leq N: M^nx + cn +d \in \overline{S}$. Observe that if for all $n \leq N$ the point $M^nx + cn + d$ is not in $\overline{S}$, then by compactness the smallest distance from $\set{ M^nx + cn + d \mid n \leq N}$ to $\overline{S}$ is positive and hence for sufficiently small $\epsilon$ the target $S$ cannot be $\epsilon$-pseudo-reached within the first $N$ steps. Conversely, if $M^nx + cn +d \in \overline{S} $ for some $n \leq N$, then because $\B(n)$ is full dimensional and contains $\zerovec$ in its interior, it follows that $(M^nx + cn +d + \epsilon\B(n) )\cap S \neq \emptyset$ for all $\epsilon > 0$. Therefore, in Case~2 pseudo-reachability can be decided by simply checking if $\set{ M^nx + cn + d \mid n \leq N}$ reaches $\overline{S}$.

Next we will show that $S$ is pseudo-reachable if Case~1 holds. 
Given $z \in \mathcal{T}$, we define a ``localisation'' of the abstraction at the point $z$ as $\aepsilon(n)(z) := D^nz + cn +d + \epsilon \B(n)$. 
Observe that $\aepsilon(n) = \{\aepsilon(n)(z) : z \in \T\}$.
This definition of a localisation will allow us to select a ``concrete trajectory'' from the set of all possible (abstract) trajectories.

Fix $\epsilon>0$ and let $T_n := \{z \in \T: \aepsilon(n)(z) \textrm{ intersects } S\}$. The next lemma implies that the sequence $T_n$ must tend towards a limiting shape; i.e. it cannot ``jump around'' forever. 
\begin{lemma}
	\label{lem::limiting-shape}
	Let $T_n = \{z : \varphi(z, n, \rho_1^n, \ldots, \rho_d^n)\}$, where $\varphi$ is a semialgebraic predicate and $\rho_1, \ldots, \rho_d$ are real algebraic, be a family of non-empty sets contained in a compact set $\T$. There exists a non-empty limiting set $L \subseteq \T$ to which the sequence $T_n$ converges as $n\to \infty$, in the following sense.
	\begin{enumerate}[a]
		\item For every $\epsilon>0$, there exists $N$ such that for all $n > N$, $T_n \subseteq L + B(\zerovec,\epsilon)$.
		\item For all $z \in L$ and $\epsilon>0$ there exists $N$ such that for all $n > N$, $z+B(\zerovec,\epsilon)$ intersects $T_n$.
	\end{enumerate}
\end{lemma}
\begin{proof}
	Write $\varphi(z, n, \rho_1^n,\ldots, \rho_d^n) = \bigvee_{\alpha \in A} \bigwedge_{\beta \in B} p_{\alpha, \beta}(z, n, \rho_1^n, \ldots, \rho_d^n) \bowtie_{\alpha, \beta} 0$. We can define the sequence $\langle T_t \mid t \in \reals\rangle$ as $T_t = \{z: \varphi(z, t, \rho_1^t, \ldots, \rho_d^t)\}$. Let $L = \{x: \liminf d(x, T_t) = 0\}$ where $d(x, T_t)$ denotes the shortest Euclidean distance from $x$ to a point in $T_t$. 
	
	We prove the first claim by contradiction. 
	Suppose there exists $\epsilon>0$ such that at infinitely many unbounded time steps $t_1 < t_2 < \ldots$ there are points $z_1, z_2, \ldots$ such that $z_i \in T_i$ but $z_i \notin L + B(\zerovec,\epsilon)$. 
	Then the sequence $z_i$ must have an accumulation point $z$ in $\T \setminus L$. But $z$ will also satisfy $\liminf d(z, T_t) = 0$ and hence $z \in L$, a contradiction.  
	
	We prove the second claim using $o$-minimality of $\reals_{\exp}$. Fix $z \in L$ and $\epsilon > 0$ and consider the set $Z = \{t \in \reals: z + B(\zerovec,\epsilon) \textrm{ intersects } T_t\}$. The set $Z$ is $o$-minimal, and since $z \in L$, it is unbounded from above. Hence it must contain an interval of the form $(N, \infty)$, which implies the desired result.
\end{proof}	
One can also show that the set $L$ described above is in fact semialgebraic, but this is not necessary for our arguments. We are now ready to show that $S$ is pseudo-reachable if Case~1 of \autoref{lem::discrete-dichotomy} holds.
\begin{lemma}
	If for every $\epsilon>0$ there exists $N_\epsilon$ such that for all $n > N_\epsilon$, $\aepsilon(n)$ intersects $S$ then $S$ is pseudo-reachable.
\end{lemma}
The main idea of the proof is to use the assumption that $\mathcal{A}_{\epsilon/2}(n)$ intersects $S$ for sufficiently large $n$ to construct an $\epsilon$-pseudo-orbit that hits $S$. 
%We will do this by showing that there exist infinitely many $n$ such that $\PO_\epsilon(n)\supseteq \mathcal{A}_{\epsilon/2}(n)$.
%That is, $S$ can be reached via an $\epsilon$-pseudo-orbit at arbitrarily large time steps.
Intuitively, in order to simulate $\mathcal{A}_{\epsilon/2}(n)$ using an $\epsilon$-pseudo-orbit, $\epsilon/2$ of the total control allowance is used to replicate the effect of the control inputs (of size at most $\epsilon/2$, corresponding to the $\frac{\epsilon}{2}\B(n)$ term in the definition of $\mathcal{A}_{\epsilon/2}(n)$) and the remaining $\epsilon/2$ is used to compensate for the abstraction from the starting point $a$ to the set $\T$. In fact, we do not know if one can deduce that $S$ is $\epsilon$-pseudo-reachable from knowing that $\aepsilon(n) \textrm{ intersects } S$ for sufficiently large $n$. This illustrates the reason why the pseudo-reachability problem is easier than the $\epsilon$-pseudo-reachability problem; see \autoref{hardness} for a more concrete argument. 
\begin{proof}
	Fix $\epsilon>0$. We show how to construct an $\epsilon$-pseudo-orbit that hits $S$. Consider~$\mathcal{A}_{\epsilon/2}(n)$. 
	By assumption, there exists $N_1$ such that for all $n > N_1$, $\mathcal{A}_{\epsilon/2}(n)$ intersects $S$. 
	We now investigate which localisations of the abstraction are responsible for intersecting $S$.
	Apply \autoref{lem::limiting-shape} to the sequence of sets $T_n = \{z \in \mathcal{T} : \mathcal{A}_{\epsilon/2}(n)(z) \textrm{ intersects } S\}$ to obtain their ``limit'' $L$. Fix any $p \in L$.
	
	Let  $\epsilon'$ be small enough so that $\epsilon' D^n \B \subseteq \frac{\epsilon}{2}\B(n)$ for all $n > 0$. 
	Intuitively, such $\epsilon'$ must exist because $D^{n}\B$ and $D^{n-1} \B$ only differ by at most a constant factor that only depends on the magnitudes $\rho_1, \ldots, \rho_d$ of eigenvalues of $M$, and we have that $D^{n-1}\B \subseteq \sum_{i=0}^{n-1}D^i\B  = \B(n)$.
	By \autoref{lem::limiting-shape}~(b), there exists $N>N_1$ such that for all $n > N$, $p+B(\zerovec, \epsilon'/2)$ intersects~$T_n$.
	That is, for all $n > N$ there exists $p_n \in \T$ such that $||p - p_n|| < \epsilon'/2 $ and $p_n \in T_n$. Equivalently,
	\[
	||p - p_n|| < \epsilon'/2 \textrm{ and } \mathcal{A}_{\epsilon/2}(n)(p_n) \textrm{ intersects } S.
	\]
	
	By Kronecker's theorem there must exist $m > N$ such that $||R^{m}x - p|| < \epsilon'/2$. 
	Hence we have $||R^mx - p_m|| < \epsilon'$ which implies $R^mx - p_m \in \epsilon' \B$ and hence $D^m(R^mx - p_m) \in \epsilon'D^m\B$. 
	Since by construction of $\epsilon'$ we have $\epsilon'D^m \B \subseteq \frac{\epsilon}{2}\B(m)$, it follows that $D^m(R^mx - p_m) \in \frac{\epsilon}{2}\B(m)$ and hence $D^mp_m \in D^m R^m x+\frac{\epsilon}{2}\B(m)$. Therefore,
	\[
	\PO_\epsilon(m) = (D^m R^m x+\frac{\epsilon}{2}\B(m))  + cm + d + \frac{\epsilon}{2}\B(m) \supseteq D^mp_m + cm +d + \frac{\epsilon}{2}\B(m) = \mathcal{A}_{\epsilon/2}(m)(p_m).
	\]
	Since $\mathcal{A}_{\epsilon/2}(m)(p_m)$ intersects $S$, it then follows that $\PO_\epsilon(m)$ too must intersect $S$. 
\end{proof}

\subsection{The algorithm}
To summarise, the analysis above gives us the following algorithm for determining if $S$ is pseudo-reachable, i.e. if $\forall \epsilon>0. \, \exists n: \PO_\epsilon(n) \cap S \neq \emptyset$. 
Let $\varphi(n, \epsilon)$ be a quantifier-free formula in $\reals_{\exp}$ defining the abstraction $\aepsilon(n)$. First determine, using the algorithm described in the proof of \autoref{lem::discrete-dichotomy}, whether Case~1 or Case~2 holds. If the former holds, then conclude that $S$ is pseudo-reachable. If Case~2 holds, then compute the value of $N$ effectively and check if there exists $n < N$ such that $(M^nx + cn + d) \cap \overline{S} \neq \emptyset$.

	% !TEX root = main.tex
\section{Skolem-hardness of the $\epsilon$-pseudo-reachability problem}
\label{hardness}
In this section we consider the $\epsilon$-pseudo-reachability problem for discrete diagonalisable systems: given diagonalisable $M$, starting point $s$, a target set $S$ and $\epsilon>0$, decide whether there exists $n$ such that $M^ns + \sum_{k=0}^{n-1}M^kB(\zerovec, \epsilon) \cap S \neq 0$. 
This problem is also known as the reachability problem for linear time-invariant systems \cite{hscc19} with the control set $B(\zerovec, \epsilon)$. 
We will reduce a hard subclass of the Skolem problem to our $\epsilon$-pseudo-reachability problem.
%We will show the hardness of this problem where $E$ is the unit $\ell_2$-ball. %With minor modifications we can also show the hardness of the version where $E = \B = \B_M$.

The Skolem problem is not known to be decidable for orders $d \geq 5$, even for diagonalisable recurrences. 
The largest class of sequences for which decidability is known is the MSTV (Mignotte-Shorey-Tijdeman-Vereschagin) class, which consists of all linear recurrence sequences over integers that (i) have at most three dominant roots with respect to the usual (Archimedean) absolute or (ii) have at most two dominant roots with respect to a $p$-adic absolute value
%or (iii) are degenerate and of order at most five 
\cite{lics22}. 
We consider the Skolem problem for integer sequences whose roots $\rho, \lambda_1, \ldots, \lambda_d$ satisfy $\rho = |\lambda_1|= \cdots = |\lambda_d|$. 
This class of sequences contains many instances that are not in the MSTV class and hence is a hard subclass of the Skolem problem.
%Observe that for diagonalisable systems there is no Diophantine hardness; hence it is expected that we have to use diagonalisable lrs.

Recall that any linear recurrence sequence can be written as $u_n = c^\top M^n s$ where $M$ is the companion matrix of $u_n$ whose eigenvalues are the roots of $u_n$. Let $u_n$ be a diagonalisable sequence that belongs to the hard subclass described above, i.e. $u_n  = c^\top M^n s$ where  $M = \operatorname{diag}(\Lambda_1, \ldots, \Lambda_d, \rho)$ and $\Lambda_i$ is a $2 \times 2$ real Jordan block with $\rho(\Lambda_i) = \rho$ for $1 \leq i \leq d$. 
We reduce the problem ``does $u_n$ have a zero?'' to an $\epsilon$-pseudo-reachability problem. 

%Write $s = (s_1, \ldots, s_d, s_{d+1})$ and $c = (c_1, \ldots, c_d, c_{d+1})$ where $s_k, c_k \in \reals^2$ for $1 \leq k \leq d$ and $s_{d+1}, c_{d+1} \in \reals$. 
%Then $u_n = \sum_{i=1}^d c_i \Lambda_i^n s_i + c_{d+1}\rho^ns_{d+1}$.
%Wlog we can assume that $c_k, s_k \neq \zerovec$ for $1 \leq k \leq d+1$, as otherwise $u_n$ would be equilvalent to a sequence of a lower order.

Consider the sequence $v_n=u_n^2$. Observe that $v_n = \sum_{i=1}^L c_i \Gamma_i^n s_i + Cr^n$ where
\begin{itemize}
	\item $r=\rho^2$,
	\item $\Gamma_i$ is a $2 \times 2$ real Jordan block with $\rho(\Gamma_i) = r$ for $1 \leq i \leq L$,
	\item $c_i, s_i \in \reals^2$ for $1 \leq i \leq L$, and
	\item $C > 0$.
\end{itemize}
The first two statements follow from the fact that the eigenvalues of $v_n$ are products of eigenvalues of $u_n$.
That $C > 0$ can be deduced as follows. Consider $w_n = \sum_{i=1}^L c_i \Gamma_i^n s_i $. It only has non-real roots and hence by \cite{nagasaka1990asymptotic} is infinitely often positive and negative. Hence if $C$ is not positive, then $v_n < 0$ for infinitely many $n$, which contradicts the fact that $v_n \geq 0$.

Next observe that $u_n$ has a zero iff there exists $n$ such that $v_n \leq 0$. 
Since we are interested  only in the sign of $v_n$, by scaling $v_n$ by $C(2r)^n$ if necessary we assume that $r \in (0,1)$ and $C = 1$.
We will construct an instance of the $\epsilon$-pseudo-reachability problem that is positive if and only if there exists $n$ such that $v_n \leq 0$.

Define
\begin{itemize}
	\item $A = \operatorname{diag}(\Gamma_1, \ldots, \Gamma_L)$,
	\item $s = (s_1, \ldots, s_L)$ and $c = (c_1, \ldots, c_L)$,
	\item $\epsilon = \frac{1-r}{||c||}$, and
	\item $H = \{z :c^\top \cdot z + 1 \leq 0\}$.
\end{itemize}
Observe that $H$ is $\epsilon$-pseudo-reachable if and only if $A^ns + \sum_{i=0}^{n-1}A^iB(\zerovec, \epsilon) \cap H \neq 0$ for some~$n$. 
Since $A^iB(\zerovec, \epsilon) = B(\zerovec, r^i \epsilon)$,  we have $\sum_{i=0}^{n-1}A^iB(\zerovec, \epsilon)  = B(\zerovec, \frac{1-r^n}{1-r}\epsilon) \coloneqq \B(n)$ and
\[
\textrm{$H$ is $\epsilon$-pseudo-reachable } \iff
\min_{z \in \B(n)} c \cdot (A^n s + z) +1\textrm{ is $\leq 0$ for some $n$.}
\]
We will show that in fact $\min_{z \in \B(n)} c \cdot (A^n s + z) +1 = v_n$, which will conclude the proof.
\begin{equation*}
\begin{split}
\min_{z \in \B(n)} c \cdot (A^n s + z) + 1  & = \sum_{i=1}^L c_i \Gamma_i^n s_i + 1 +  \min_{z \in \B(n)} c \cdot z \\
&= \sum_{i=1}^L c_i \Gamma_i^n s_i + 1  - ||c|| \frac{1-r^n}{1-r}\epsilon\\
&=\sum_{i=1}^L c_i \Gamma_i^n s_i  + r^n \\&= v_n.
\end{split}
\end{equation*}
	
	% !TEX root = main.tex
\section{The continuous-time pseudo-reachability problem}\label{sec:cont_pseudo_reach}

In this section we show that the approach we described in \autoref{sec::disc-decidability} for deciding the discrete-time pseudo-reachability problem for diagonalisable systems also works in the continuous setting with one important difference: to handle Case~2 of the dichotomy lemma (exactly the same as \autoref{lem::discrete-dichotomy}) we need to solve the \emph{bounded-time reachability problem for continuous-time affine dynamical systems}, which is only known to be decidable assuming Schanuel's conjecture~\cite{contSkolem}.
For detailed proofs see the full version of the paper.

Let $M = \operatorname{diag}(\Lambda_1, \ldots, \Lambda_k, \rho_{k+1}, \ldots, \rho_{d})\in(\reals\cap\algebraics)^{L\times L}$ be a diagonalisable matrix in real Jordan form, $s\in \rats^{L}$ be a starting point, $b\in \rats^{L}$ be an affine term and $S\subseteq \reals^{L}$ be a semialgebraic target set. 
The trajectory of the system (in the absence of additional control inputs) is given by 
\[
x(t)=e^{Mt}s + \int_{0}^{t} e^{Mh}b \,dh.
\]
Intuitively, while in the discrete setting control inputs are applied after each unit of time and thus are represented by a sequence $\langle d_n \mid n\in \nats\rangle$, in the continuous setting they are represented by a continuous function $\Delta: \reals_{\geq 0} \to \reals^{L}$. Hence an $\epsilon$-pseudo-orbit is defined as a trajectory
\[
x(t)=e^{Mt}s + \int_{0}^{t} e^{Mh}b \,\, dh + \int_{0}^{t} e^{Mh}\Delta(t-h) \, dh.
\]
for some control signal $\Delta: \reals_{\geq 0} \to \reals^L$ satisfying $||\Delta_\epsilon(t)|| \le \epsilon$ for all $t \geq 0$.
The pseudo-reachability problem is then defined in the same way as before: decide whether for every $\epsilon>0$ there exists an $\epsilon$-pseudo-orbit that reaches $S$.

Let $\B$ be the same control set as defined in \autoref{sec::closed-form}.
For $1 \leq j \leq k$, let $r_j = \operatorname{Re}(\lambda_j)$ and $\omega_j = \operatorname{Im}(\lambda_j)$  where $\lambda_j$ is a non-real eigenvalue of the block $\Lambda_j$.
For $k < j \leq d$ let $r_j = \rho_j$ and $\omega_j = 0$.
%a full dimensional shape constructed by the following ``product of $\ell_2$-balls''
%$$
%\B = \Pi_{i=1}^k \ball((0,0),1) \times \Pi_{i=k+1}^{d}\ball(0,1),
%$$
%where $\ball(c,1)$ represents a $\ell_2$ ball centered at $c$ and of radius $1$ \MS{we should define the control set once and forever}. 
By using essentially the same arguments as in \autoref{sec::closed-form}, we can show that the pseudo-reachability problem is equivalent to determining the truth of
\[
\forall \epsilon>0. \, \exists t: (e^{Mt}x + ct + d + \epsilon\B(t)) \cap S \neq \emptyset
\]
where $x, c, d$ are $L$-dimensional vectors and $\B(t) = \{z : \varphi(z,t, e^{r_1t}, \ldots, e^{r_dt})\}$ for semialgebraic predicate $\varphi$. We denote the term $e^{Mt}x + ct + d + \B(t)$ by $\PO_{\epsilon}(n)$.

To define a convenient abstraction, we again write $e^{Mt} = D(t)R(t)$ where
$D(t)\coloneqq \operatorname{diag}(e^{r_1t}, e^{r_1t},\dots,e^{r_kt},e^{r_kt},e^{r_{k+1}t},e^{r_{k+2}t},\dots,e^{r_{k+d}t})$ is diagonal and
$R(t)$ is a block diagonal matrix whose blocks are rotation matrices of the form
$\begin{bmatrix}\cos(\omega_jt)&-\sin(\omega_it)\\\sin(\omega_jt)&\cos(\omega_jt)\end{bmatrix}$ for $1 \leq j \leq k$ and are of the form $\Omega_i=\begin{bmatrix}1\end{bmatrix}$ for $k+1\leq j \leq d$.
Just as in the discrete case, we next define 
\[
\mathcal{T} \coloneqq \operatorname{cl}( \{R(t)x : t \in \reals_{\geq 0}\})
\textrm{ and }
\aepsilon (t) \coloneqq D(t)\mathcal{T} + ct +d + \epsilon \B(t),
\]
where $\mathcal{T}$ is again semialgebraic and effectively computable \cite{contSkolem}
and $\aepsilon$ acts as an abstraction of $\PO_\epsilon$. In particular, for all $\epsilon>0$ and $t\in \reals_{\geq 0}$, we have $\PO_{\epsilon}(t)\subseteq\aepsilon(t)$. Moreover, observe that $\aepsilon(t) = \varphi(t, e^{r_1t}, \ldots, e^{r_dt})$ for a semialgebraic function $\varphi$, which is the most important property we need. 
%Given $z \in \mathcal{T}$, we define $\aepsilon(t)(z) = D(t)z + ct +d + \epsilon \B(t)$.  
We use $\aepsilon(t)$ in the same way we used $\aepsilon(n)$ in the discrete case to arrive at the following dichotomy lemma.
\begin{lemma}
	\label{lem::cont-dichotomy}
	Either
	\begin{enumerate}
		\item for every $\epsilon > 0$ there exists $T_\epsilon$ such that for all $t > T_\epsilon$, $\aepsilon(t)$ intersects $S$, or
		\item there exist $T$ and $\epsilon>0$, both effectively computable, such that $\aepsilon(t)$ does not intersect $S$ for all $t > T$. 
	\end{enumerate}
	Moreover, it can be effectively determined which case holds.
\end{lemma}
\begin{proof} 
	First we show that the dichotomy holds, putting the issues of effectiveness aside. 
	Let 
	\[
	\Phi(\epsilon, t) = \bigvee_{\alpha \in A} \bigwedge_{\beta \in B} p_{\alpha, \beta}(\epsilon,t,e^{r_1t}, \ldots, e^{r_dt}) \bowtie_{\alpha, \beta} 0
	\]
	be a quantifier-free formula equivalent to $\aepsilon(n) \cap S \ne \emptyset$. 
	Suppose Case~1 does not hold. Then there exists a particular $\epsilon>0$ such that $\Phi(\epsilon, t) $ does not hold for arbitrarily large~$t$. Consider the set $\{t \geq 0: \Phi(\epsilon,t) \textrm{ does not hold}\}$. By $o$-minimality of $\realexp$ this set is a finite union of intervals and since it contains arbitrarily large real numbers by assumption, it must contain an unbounded interval $(T, \infty)$.
	That is, for all $t>T$ the formula $\Phi(\epsilon, t)$ does not hold.
	
	\proofsubparagraph{Effectiveness.}We now address the issues of effectiveness. Consider the formula 
	\[
	\Psi(\epsilon) = \exists T_\epsilon. \, \forall t > T_\epsilon: \Phi(\epsilon, t).
	\]
	We show that $\Psi(\epsilon)$ is equivalent to a formula $\psi(\epsilon)$ in the language of $\reals_0$. To determine which case holds it then remains to determine the truth value of the sentence $\forall \epsilon: \psi(\epsilon)$.
	
	By the $o$-minimality argument above, given $\epsilon>0$, each the values of $t$ for which $p_{\alpha, \beta}(\epsilon,t,e^{r_1t}, \ldots, e^{r_dt}) \bowtie_{\alpha, \beta} 0$ holds is either bounded or contains an unbounded interval. 
	By elementary considerations it follows that $\Psi(\epsilon)$ is equivalent to 
	\[
	\bigvee_{\alpha \in A} \bigwedge_{\beta \in B} \, \exists T_\epsilon. \, \forall t > T_\epsilon:  p_{\alpha, \beta}(\epsilon,t,e^{r_1t}, \ldots, e^{r_dt}) \bowtie_{\alpha, \beta} 0.
	\]
	Hence it suffices to show how to construct a formula in the language of $\reals_0$  that is equivalent to $\exists T_\epsilon. \, \forall t > T_\epsilon:  p_{\alpha, \beta}(\epsilon,t,e^{r_1t}, \ldots, e^{r_dt}) \bowtie_{\alpha, \beta} 0$.
	For each $\epsilon>0$, the formula first tests if $p_{\alpha, \beta}(\epsilon)$ (as a polynomial in $d+1$ remaining variables) is identically zero.
	If yes, then $\varphi(\epsilon)$ is $\textbf{true}$ or $\textbf{false}$ depending only on $\bowtie_{\alpha, \beta}$.
	Otherwise, write $p_{\alpha, \beta}(\epsilon,t,e^{r_1t}, \ldots, e^{r_dt}) = \sum_{i=1}^k q_i(\epsilon, n)R_i^n$ where $q_i(\epsilon)$ is not identically zero for all $i$ and $R_1 > \cdots > R_k > 0$ are of the form $e^{p_1r_1+\ldots+p_dr_d}$ for $p_1, \ldots, p_d \in \nats$. 
	Since $\left|q_1(\epsilon, n)R_1^n\right|> \left| \sum_{i=2}^k q_i(\epsilon, n)R_i^n \right|$ for sufficiently large $n$, whether $p_{\alpha, \beta}(\epsilon,n, \rho_1^n, \ldots, \rho_d^n) \bowtie_{\alpha, \beta} 0$ holds for sufficiently large $n$ depends only on $q_1(\epsilon, n)$. 
	Hence we can choose $\psi(\epsilon)$ to be $\lim_{n\to \infty} q_1(\epsilon) \bowtie_{\alpha, \beta} 0$, which amounts to a sign condition on the coefficients of $q_1(\epsilon, n)$. 
	
	\proofsubparagraph{Computing $T$.} Finally, we show that in Case~2, the value $T$ can be effectively computed. To this end, by repeatedly trying smaller and smaller values of $\epsilon$ first compute a rational $e>0$ such that $\Psi(e)$ (equivalently, $\psi(e)$) does not hold. 
	To be able to compute $T$ it then suffices to compute, for a particular $(\alpha, \beta)$, a value $T_{\alpha,\beta}$ such that $p_{\alpha, \beta}(\epsilon,t,e^{r_1t}, \ldots, e^{r_dt}) \bowtie_{\alpha, \beta} 0$ does not hold for all $t > T_{\alpha,\beta}$, assuming that it does not hold for sufficiently large~$t$. We can then take $T$ to be the maximum of $T_{\alpha,\beta}$ over $(\alpha,\beta) \in A\times B$.
	
	To compute $T_{\alpha, \beta}$, consider $p \coloneqq p_{\alpha, \beta}(e)$. Assuming it is not identically zero (otherwise we can choose $T_{\alpha, \beta}$ to be any positive integer), write  $p(t,e^{r_1t}, \ldots, e^{r_dt}) = \sum_{i=1}^k q_i(n)R_i^n$ where $q_i$ is not identically zero for all $i$ and $R_1 > \cdots > R_k > 0$ are of the form $e^{p_1r_1+\ldots+p_dr_d}$ for $p_1, \ldots, p_d \in \nats$. 
	Since $p_{\alpha, \beta}(\epsilon,t,e^{r_1t}, \ldots, e^{r_dt})  \bowtie_{\alpha, \beta} 0$ and hence $p(t,e^{r_1t}, \ldots, e^{r_dt})  \bowtie_{\alpha, \beta} 0$ do not hold for sufficiently large $n$, it must be the case that $q_1(n) \bowtie_{\alpha, \beta} 0$ does not hold for sufficiently large $n$. 
	Hence it remains to choose $N_{\alpha, \beta}$ large enough so that for all $n > N_{\alpha, \beta}$, $|q_1(n)R_1^n|$ dominates $\left| \sum_{i=2}^k q_i(n)R_i^n \right|$.
\end{proof}

We next show pseudo-reachability in Case~1.
\begin{lemma}
	If for every $\epsilon>0$ there exists $T_\epsilon$ such that for all $t > T_\epsilon$, $\aepsilon(t)$ intersects $S$ then $S$ is pseudo-reachable.
\end{lemma}
\begin{proof}
	We first define a suitable notion of localisation, in exacty the same way as the discrete case. 
	Given $z \in \torus$, let $\aepsilon(t)(z) = D(t)z + ct +d +\epsilon \B(t)$.
	
	Fix $\epsilon>0$. We show how to construct an $\epsilon$-pseudo-orbit that hits $S$.
	In fact, this $\epsilon$-pseudo-orbit will hit $S$ at an integer time step $m$. Consider~$\mathcal{A}_{\epsilon/2}(t)$. 
	By assumption, there exists $T_1$ such that for all $t > T_1$, $\mathcal{A}_{\epsilon/2}(t)$ intersects $S$. 
	We now investigate which localisations of the abstraction are responsible for intersecting $S$.
	Apply \autoref{lem::limiting-shape} to the sets $T_n = \{z \in \mathcal{T} : \mathcal{A}_{\epsilon/2}(n)(z) \textrm{ intersects } S\}$, $n \in \nats$, to obtain their `limit' $L$. Fix any $p \in L$.
	
	Let  $\epsilon'$ be small enough so that for all $n > 0$, $\epsilon' D(n)B(\zerovec, 1) \subseteq \frac{\epsilon}{2}\B(n)$.
	By \autoref{lem::limiting-shape}~(b), there exists $N>T_1$ such that for all integers $n > N$, $p+B(\zerovec, \epsilon'/2)$ intersects $T_n$.
	That is, for all $n > N$ there exists $p_n \in \T$ such that $||p - p_n|| < \epsilon'/2 $ and $p_n \in T_n$. Equivalently,
	\[
	||p - p_n|| < \epsilon'/2 \textrm{ and } \mathcal{A}_{\epsilon/2}(n)(p_n) \textrm{ intersects } S.
	\]
	
	By Kronecker's theorem there must exist integer $m > N$ such that $||R(m)x - p|| < \epsilon'/2$. 
	Hence we have $||R(m)x - p_m|| < \epsilon'$ which implies $R(m)x - p_m \in \epsilon' B(\zerovec, 1)$ and hence $D(m)(R(m)x - p_m) \in \epsilon'D(m)B(\zerovec, 1)$. 
	Since by construction of $\epsilon'$ we have $\epsilon'D(m)\B(\zerovec, 1) \subseteq \frac{\epsilon}{2}\B(m)$, it follows that $D(m)(R(m)x - p_m) \in \frac{\epsilon}{2}\B(m)$ and hence $D(m)p_m \in D(m) R(m) x+\frac{\epsilon}{2}\B(m)$. Therefore,
	\begin{align*}
		\PO_\epsilon(m) &= \left(D(m) R(m) x+\frac{\epsilon}{2}\B(m)\right)  + cm + d + \frac{\epsilon}{2}\B(m)\\
		& \supseteq D(m)p_m + cm +d + \frac{\epsilon}{2}\B(m)\\
		&= \mathcal{A}_{\epsilon/2}(m)(p_m).
	\end{align*}
	Since $\mathcal{A}_{\epsilon/2}(m)(p_m)$ intersects $S$, it then follows that $\PO_\epsilon(m)$ too must intersect $S$. 
\end{proof}

We thus have reduced the pseudo-reachability problem to that of handling Case~2:

\begin{theorem}
	The continuous-time pseudo-reachability problem reduces to the bounded-time reachability problem for continuous-time affine dynamical systems.
\end{theorem}
Intuitively, the dichotomy lemma (\autoref{lem::discrete-dichotomy}) holds verbatim for the continuous systems, and in Case~1 again $S$ is always pseudo-reachable. It then remains to handle Case~2.
Since bounded-time reachability problem for continuous-time affine dynamical systems can be encoded in $\reals_{\exp, \cos \restr [0, T]}$, we have the following (conditional) decidability result.
\begin{corollary}
	Continuous-time pseudo-reachability problem for diagonalisable affine dynamical systems is decidable subject to Schanuel's conjecture.
\end{corollary}

\section{Discussion}

%Difference to MFCS21: Intuitively, already for a target $S$ that is an intersection of two halfspaces, deciding pseudo-reachability becomes difficult because given $\epsilon>0$, there is no compact description of all $\epsilon$-pseudo-orbits that reach each of the halfspaces that would allow one to check whether there exists a single $\epsilon$-pseudo-orbit that reaches both halfspaces at the same time.

The main technical result of our paper is that it is decidable whether 
\[
\forall \epsilon > 0.\, \exists n: \, \left(M^nx + f(n) + \epsilon\B(n)\right) \cap S \neq 0
\]
where $M$ is a diagonalisable matrix with algebraic entries,
$x$ is an algebraic starting point,
$f$ is a semialgebraic function,
$S$ is a semialgebraic target and
$\B(n)=\varphi(n, \rho_1^n, \ldots, \rho_d^n)$ for $\rho_1, \ldots, \rho_d \in \reals \cap \algebraics$
and a semialgebraic function $\varphi$. 
We used this result to show decidability of the discrete-time pseudo-reachability problem for diagonalisable systems in the following way.
We first observed that the pseudo-reachability problem can be cast as the problem of determining whether $\forall \epsilon>0. \, \exists n: \PO_{\epsilon}(n) \cap S \neq \emptyset$, where $\PO_\epsilon(n)$ is the set of all points that are reachable exactly at the time $n$ via an $\epsilon$-pseudo-orbit.
After choosing $\B$ as the most convenient control set (see \autoref{convenient-control-set} and \autoref{sec::closed-form}), we then showed that $\PO_\epsilon(n)$ can be written as $M^nx + f(n) + \epsilon\B(n)$.

The reason we are unable to show decidability for non-diagonalisable systems in this fashion is that we are unable to write $\PO_{\epsilon}(n)$ as $M^nx + f(n) + \epsilon\B(n)$.
For example, already for the Jordan block $M=\begin{bmatrix}
	1 & 1\\
	0 & 1
\end{bmatrix}$, and in general already for blocks with a single repeated real eigenvalue, we do not know whether it is even possible to eliminate the summation $\sum_{i=0}^{n-1}M^iB$ and express $\PO_\epsilon(n) = M^nx + \epsilon \sum_{i=0}^{n-1}M^iB$, where $B$ is any full dimensional shape containing $\zerovec$ in its interior, in the required fashion.
%The approach developed by D'Costa et al. where they show decidability of the Pseudo-Skolem Problem for all systems is not applicable either. Intuitively, already for a target $S$ that is an intersection of two halfspaces, deciding pseudo-reachability becomes difficult because given $\epsilon>0$, there is no compact description of all $\epsilon$-pseudo-orbits that reach each of the halfspaces that would allow one to check whether there exists a single $\epsilon$-pseudo-orbit that reaches both halfspaces at the same time.

Our approach, however, can be used to solve, in full generality, the robre	ust reachability problem of \cite{akshay_et_al:LIPIcs.STACS.2022.5}:
given $M, x$ and $S$, decide whether $\forall \epsilon > 0: \, \exists n: \, (M^nx + \epsilon M^nB(\zerovec,1)) \cap S \neq \emptyset$.
Intuitively, the reason is that in this version there is no summation of the form $\sum_{i=0}^{n-1}M^iB$.
Detailed proofs (for both the discrete-time and the continuous-time versions) can be found in the appendix.
For diagonalisable systems in particular, decidability of the robust reachability problem is almost immediate.
First, one can again show that the problem is equivalent to determining whether $\forall \epsilon > 0: \, \exists n: \, (M^nx + \epsilon M^n\B) \cap S \neq \emptyset$.
It then remains to observe that $M^n\B = \varphi(n, \rho_1^n, \ldots, \rho_d^n)$ for a semialgebraic predicate $\varphi$ and apply the technical result described above.

\bibliography{references.bib}
\appendix
\section{Discrete robust reachability problem}\label{sec::discrete_robust_reach}

In this section, we show the full decidability (including for non-diagonalisable systems) of the discrete-time robust reachability problem: decide, given $M \in (\reals \cap \algebraics)^{L\times L}$, a starting point $x\in\rats^L$ and a target $S\subseteq \reals^L$, whether for every $\epsilon>0$ there exists $n$ and $\delta \in B(0, \epsilon) = \epsilon B(\zerovec, 1)$ %\MS{(This seems confusing. We should define $B(\cdot,\cdot)$, and $\boldsymbol{1}$, $\boldsymbol{0}$ in preliminaries.)} 
such that $M^n(x+\delta)\in S$. As discussed in \autoref{sec::disc-decidability}, wlog we can assume that $M$ is in real Jordan form: \[
M = \operatorname{diag}(J_1, \ldots J_k, J_{k+1}, \ldots, J_{d})
\]
where for $1\leq i \leq k$ the block $J_i$ has two non-real eigenvalues and for $k<i\leq d$ the block $J_i$ has one real eigenvalue. We denote the multiplicity and the spectral radius of $J_i$ by $\sigma_i$ and~$\rho_i$, respectively. 

As discussed in \autoref{sec::disc-decidability}, the robust reachability problem can be equivalently stated in terms of any full-dimensional set $\B$ that contains $\zerovec$ in its interior (instead of $B(\zerovec,1)$) as the ``control set''.  
That is, for any such set $\B$, the problem of deciding whether $\forall \epsilon.\,  \exists n : (M^nx + \epsilon M^n \B )\cap S \neq \emptyset$ is equivalent to the robust reachability problem.
We first give a set $\B$ that is most appropriate for our purposes. 
Intuitively, the idea is again to eliminate the rotations in $M^n$ so that $M^n\B$ can be defined in a first-order fashion using algebraic parameters.
Let $\B = \Pi_{i=1}^d \B_i$ where (i)~$\B_i = \Pi_{j=1}^{\sigma_i} B((0,0), 1)$ for $1\leq i \leq k$ and (ii) $B_i = [-1,1]^{\sigma_i}$ for $k< i \leq d$. 
Define $\B(n) = M^n\B$ %\MS{or $M^{n-1}\B$??}
 and observe that 
\[
\B(n) = \varphi(n, \rho_1^n, \ldots, \rho_d^n)
\]
where $\varphi$ is a semialgebraic function.
Let $\PO_\epsilon(n)=M^nx + \epsilon\B(n)$. The robust reachability problem is then equivalent to determining whether
 \[
\forall  \epsilon>0. \, \exists n:\, \PO_\epsilon(n) \cap S \ne \emptyset. 
 \]

We move onto defining the abstraction for $M^nx$. 
Assume $M$ is of the same form as above.
For $\alpha\in\torus$ let $
R(\alpha) = \begin{bmatrix}
	\operatorname{Re}(\alpha) & -\operatorname{Im}(\alpha)\\
	\operatorname{Im}(\alpha) & \operatorname{Re}(\alpha)
\end{bmatrix}
$
and for $1 \leq i \leq k$ let $\gamma_i = \lambda_i / \rho(J_i)$ for a non-real eigenvalue $\lambda_i$ of the block~$J_i$.  %\MS{shouldn't $\gamma_i$ take a form similar to $R(\alpha)$?}.
Let $f:\nats \times \torus^k\to \reals^{L\times L}$ be the ``matrix builder'' function, defined as follows.
\begin{equation}\label{eq:mat_builder_discrete}
f(n, (\alpha_1,\ldots, \alpha_k)) = \operatorname{diag}(g(\alpha_1),\ldots,g(\alpha_k), J_{k+1}^n, \ldots, J_{d}^n),
\end{equation}
where
\[
g(\alpha_i) = \begin{bmatrix}
	A_i & n\Lambda_i^{-1}A_i & \cdots &\binom{n}{\sigma_i-1}\Lambda_i^{-\sigma_i+1}A_i\\
	& A_i & \ddots & \vdots\\
	&&\ddots&  n\Lambda_i^{-1}A_i\\
	&&&A_i 
\end{bmatrix}
\textrm{ and }
A_i = \rho_i^n R(\alpha).
\]
We define the matrix builder with respect to the state matrix $M$ which is given in real JNF. For example, for $\alpha, \beta \in \torus$ and the state matrix
\[
M =
\begin{bmatrix}
	\Lambda_1&I\\
	&\Lambda_1&\\
	&&\Lambda_2&I\\
	&&&\Lambda_2\\
	&&&&\rho_3&1\\
	&&&&&\rho_3
\end{bmatrix},\]
the corresponding matrix builder takes the form
\[
f(n, (\alpha,\beta)) =
\begin{bmatrix}
	A&n\Lambda_1^{-1}A\\
	&A&\\
	&&B&n\Lambda_2^{-1}B\\
	&&&B\\
	&&&&\rho_3^n&n\\
	&&&&&\rho_3^n
\end{bmatrix}.
\]
Here $A = \rho(\Lambda_1)^nR(\alpha)$
and 
 $B =  \rho(\Lambda_2)^nR(\beta)$.
Let
\[
\T = \operatorname{cl}(\{(\gamma_1^n,\ldots, \gamma_k^n) : n \in \nats\}).
\]
The set $\T$ is semialgebraic and effectively computable. Further define
\[
\aepsilon(n)(z)= f(n,z) x + \epsilon\B(n) \textrm{ and }  \aepsilon(n)=\{\aepsilon(n)(z) : z \in \T\}.
\]
Then $M^nx = f(n, (\gamma_1^n,\ldots, \gamma_k^n) )x$ is abstracted by $\{f(n, z)x : z \in \torus^k\}$ and $\aepsilon(n) \supseteq \PO_\epsilon(n)$.

The following lemma encapsulates all the nasty differences between the diagonalisable and the non-diagonalisable case. Its proof is an easy manipulation of matrices.
\begin{lemma}
	Given $z = (\alpha_1, \ldots, \alpha_k)$, a time step $n$, an update matrix $M$ and a starting point $x$,
	\[
	f(n,(\alpha_1, \ldots, \alpha_k))x = M^n x + M^n \Delta
	\] %\MS{shouldn't we consider $M^{n-1}\Delta$?}
	has a solution
	\[
	\Delta = (\Delta_1, \ldots, \Delta_k, \zerovec, \ldots, \zerovec)
	\]
	where $\Delta_i(j) = R(\gamma_i^{-n})(R(\alpha_i)-R(\gamma_i^n)) x_i(j)$ for $1\leq i \leq k$ and $1\leq j\leq \sigma_i$. 
\end{lemma}
Observe that, assuming $k$ and $x$ are fixed, $||\Delta|| = O(||z-\Gamma^n||)$, %\MS{what is $\Theta$?}
 where $\Gamma^n = (\gamma_1^n, \ldots, \gamma_k^n)$. Hence we obtain the following corollary, which intuitively states that if $z$ is close to $\Gamma^n$, then $f(n,z)x$ is close to the true point $M^nx$, in the sense that $f(n,z)x$ can be reached from $x$ by first jumping to a point $x'$ that is at most $\epsilon\B$ away and then applying $M$ exactly $n$ times.
\begin{corollary}
	Given $M$ and $x$, there exists $C$ such that for all $\epsilon>0$ and $z\in \torus^k$,
	\[
	||\Gamma^n-z|| < C\epsilon \implies \exists \Delta \in \epsilon\B: \, f(n,z)x = M^n x + M^n \Delta.
	\]
\end{corollary}

We now move onto proving decidability of the robust reachability problems. Firstly, the dichotomy lemma and its proof hold verbatim.	
\begin{lemma}
	Either
	\begin{enumerate}
		\item for every $\epsilon > 0$ there exists $N_\epsilon$ such that for all $n > N_\epsilon$, $\aepsilon(n)$ intersects $S$, or
		\item there exist computable $N$ and $\epsilon>0$ such that $\aepsilon(n)$ does not intersect $S$ for all $n > N$. 
	\end{enumerate}
	Moreover, it can be effectively determined which case holds.
\end{lemma}
If Case~1 holds, then $S$ is robust reachable if and only if it is reachable within the first $N$ steps. We will show that in Case~2 $S$ is robust reachable. This will conclude the proof. 
\begin{lemma}
	If for all $\epsilon>0$ there exists $N_\epsilon$ such that for all $n>N_\epsilon$, $\aepsilon(n)$ intersects $S$ then $S$ is robust reachable.
\end{lemma}
\begin{proof}
	Let $\epsilon>0$. 
	We show that $S$ is ``$\epsilon$-robust-reachable''. 
	That is, $M^nx + \epsilon M^n \B$ intersects~$S$ for some $n$.
	Consider $\mathcal{A}_{\epsilon/2}$.
	By assumption, there exists $N_1$ such that for all $n > N_1$, $\mathcal{A}_{\epsilon/2}(n)$ intersects $S$.  
	Let $\epsilon'$ be sufficiently small such that for all $z$
	\[
	||\Gamma^n-z|| < \epsilon' \implies \exists\Delta \in \frac{\epsilon}{2}\B: \, f(n,z)x = M^n x + M^n \Delta.
	\]
	Consider the limiting shape $L$ for the sequence
	\[
	T_n = \{z \in \mathcal{T} : \mathcal{A}_{\epsilon/2}(n)(z) \textrm{ intersects } S\}.
	\] 
	By \autoref{lem::limiting-shape}~(b), there exists $N>N_1$ such that for all $n > N$, $p+B(\zerovec, \epsilon'/2)$ intersects~$T_n$.
	That is, for all $n > N$ there exists $p_n \in \T$ such that $||p - p_n|| < \epsilon'/2 $ and $p_n \in T_n$. Equivalently,
	\[
	||p - p_n|| < \epsilon'/2 \textrm{ and } \mathcal{A}_{\epsilon/2}(n)(p_n) \textrm{ intersects } S.
	\]
	By Kronecker's theorem there must exist $m > N$ such that $||\Gamma^m - p|| < \epsilon'/2$. Hence we have
	\[
	||\Gamma^m-p_m||\leq \epsilon'
	\textrm{ and }
	\mathcal{A}_{\epsilon/2}(m)(p_m) \textrm{ intersects } S.
	\]
	By the construction of $\epsilon'$ there exists $\Delta \in \frac{\epsilon}{2}\B$ such that $f(n, p_m)x = M^nx + M^n\Delta$.
	Hence
	\[
	M^nx + \epsilon M^n \B = M^nx + \frac{\epsilon}{2}M^n\B + \frac{\epsilon}{2}M^n\B \supseteq f(n,p_m)x+\frac{\epsilon}{2}\B = \mathcal{A}_{\epsilon/2}(m)(p_m).
	\]
	Since $\mathcal{A}_{\epsilon/2}(m)(p_m)$ intersects $S$, it follows that $	M^nx + \epsilon M^n \B$ intersects $S$ too.
\end{proof}
\section{Continuous-time robust reachability problem}\label{sec::cont_robust_reach}
In this section, we take the continuous robust reachability problem for linear dynamical systems and show that it can be reduced into the bounded-time reachability problem for semialgebraic target sets. The techniques we use are very similar to the discrete setting. The only major difference is that deciding the continuous-time robust reachability problem, requires Schanuel's conjecture. 

Let $M\in (\reals\cap\algebraics)^{L\times L}$ be a matrix in real JNF, $x\in\rats^L$ be a starting point, and $S\subseteq\reals^L$ be a semialgebraic target set. We want to show how to decide whether $\forall \epsilon>0.\, \exists t\in\reals_{\geq 0}, \delta\in B(\boldsymbol{0},\epsilon)\colon e^{Mt}(x+\delta) \in S$.

%Let $M = \operatorname{diag}(J_1, \ldots, J_k, J_{k+1}, \ldots, J_{d})$, where the first $k$ Jordan blocks correspond to complex eigenvalues and the rest of the blocks $J_i, k+1\leq i \leq d$ correspond to the real eigenvalues of $M$. Multiplicity of $J_i$ is denoted by $\sigma_i$. 
Let $\B$ be the control set as defined for the discrete setting. This gives $\B(t)=e^{Mt}\B$ and the $\epsilon$-pseudo-orbit of the continuous system at time $t\in\reals_{\geq 0}$ can be defined as $\PO_\epsilon(t)=e^{Mt}x+\B(t)$. We can state the continuous robust reachability Problem as 
\[
\forall \epsilon>0.\;\exists t\in\reals_{\geq0}\colon\PO_{\epsilon}(t)\cap S\neq \emptyset
\]

One can define an abstraction for $e^{Mt}x$ similar to the discrete case. 
In particular, for $\alpha\in\reals\cap\algebraics$ let $
R(\alpha) = \begin{bmatrix}
	\cos(\alpha) &-\sin(\alpha)\\
	\sin(\alpha)& \cos(\alpha)
\end{bmatrix}
$
and for $1 \leq i \leq k$ let $\gamma_i =\operatorname{Im}(\lambda_i)$ and $r_i =\operatorname{Re}(\lambda_i)$
%\begin{bmatrix}
%	\cos(\omega t) &-\sin(\omega t)\\
%	\sin(\omega t)& \cos(\omega t)
%\end{bmatrix}$, where $\omega=\operatorname{Im}(\lambda_i)$ 
for a non-real eigenvalue $\lambda_i$ of the block~$J_i$.
Finally, let $f:\reals_{\geq 0} \times\reals^k\to \reals^{L\times L}$ be the ``matrix builder'' function, defined as follows.
\[
f(t, (\alpha_1,\ldots, \alpha_k)) = \operatorname{diag}(g(\alpha_1),\ldots,g(\alpha_k), e^{J_{k+1}t}, \ldots, e^{J_{d}t})
\] 
where
\[
g(\alpha_i) = \begin{bmatrix}
	A_i & tA_i & \cdots &\frac{t^{\sigma_i-1}}{(\sigma_i-1)!}A_i\\
	& A_i & \ddots & \vdots\\
	&&\ddots &  tA_i\\
	&&&A_i 
\end{bmatrix}
\textrm{ and }
A_i = e^{r_it} R(\alpha).
\]
Let
\[
\T = \{(R(\gamma_1t),\ldots, R(\gamma_kt)) : t \in \reals_{\geq 0}\}.
\]
The set $\T$ is semialgebraic and effectively computable. Further define
\[
\aepsilon(t)(z)= f(t,z) x + \epsilon\B(t) \textrm{ and }  \aepsilon(t)=\{\aepsilon(t)(z) : z \in \T\}.
\]
Notice that we have $\aepsilon(n) \supseteq \PO_\epsilon(n)$. The next lemma shows that if $z$ is picked close enough to $\Gamma(t) = (\gamma_1t, \ldots, \gamma_kt)$, then $e^{Mt}x$ can be approximated by $f(t,z)x$.
\begin{lemma}
	Given $z = (\alpha_1, \ldots, \alpha_k)$, a time point $t\in\reals_{\geq 0}$, an update matrix $M$ and a starting point $x$,
	\[
	f(t,(\alpha_1, \ldots, \alpha_k))x = e^{Mt} x + e^{Mt} \Delta
	\] 
	has a solution
	\[
	\Delta = (\Delta_1, \ldots, \Delta_k, \zerovec, \ldots, \zerovec)
	\]
	where $\Delta_i(j) = R(-\gamma_it)(R(\alpha_i)-R(\gamma_it)) x_i(j)$ for $1\leq i \leq k$ and $1\leq j\leq \sigma_i$. 
\end{lemma}
We have the following corollary, similar to the discrete setting.
\begin{corollary}
	Given $M$ and $x$, there exists $C$ such that for all $\epsilon>0$ and $z\in (\reals\cap\algebraics)^k$,
	\[
	||\Gamma(t)-z|| < C\epsilon \implies \exists \Delta \in \epsilon\B: \, f(t,z)x = e^{Mt} x + e^{Mt} \Delta.
	\]
\end{corollary}

Before stating the main result of this section, we state the dichotomy lemma for the continuous-time setting.	\begin{lemma}
	\label{lem::rob_cont-dichotomy}
	Either
	\begin{enumerate}
		\item for every $\epsilon > 0$ there exists $T_\epsilon$ such that for all $t > T_\epsilon$, $\aepsilon^c(t)$ intersects $S$, or
		\item there exist computable $T$ and $\epsilon>0$ such that $\aepsilon^c(t)$ does not intersect $S$ for all $t > T$. 
	\end{enumerate}
	Moreover, it can be effectively determined which case holds.% If Case~1 holds, $S$ is pseudo-reachable.
\end{lemma}
We first consider Case~1 and show that if this case holds, the answer to the robust reachability problem is positive. The proof would be exactly the same as for the discrete setting.
\begin{lemma}\label{lem:rob_cont_completeness}
	If for all $\epsilon>0$ there exists $T_\epsilon\in\reals_{\geq 0}$ such that for all $t>T_\epsilon$, $\aepsilon(t)$ intersects $S$, then $S$ is robust reachable.
\end{lemma}
Similar to the pseudo-reachability problem, we know that given a finite-time interval, to answer bounded-time robust reachability questions, it is enough to check whether the given target set is reachable within the specified time interval or not.  Now, we are ready to state the main result of this section.
\begin{lemma}\label{lem::robust_cont_reach_reduction}
	The continuous robust reachability problem for linear dynamics reduces to bounded-time reachability problem for linear continuous-time systems.
\end{lemma}
\begin{proof}
	Using the results of Lemma~\ref{lem::rob_cont-dichotomy}, we can effectively decide whether Case~1 holds or not. If Case~1 holds, by Lemma~\ref{lem:rob_cont_completeness}, we know that $S$ is robustly reachable. Otherwise, Case~2 holds and we effectively compute the time-bound $T$ and therefore, need to check if $S$ is reachable by the orbit of the system within the time interval $[0,T]$. Therefore, the robust reachability problem reduces into the bounded-time reachability problem for continuous-time linear dynamical systems.
	\end{proof}
Finally, since the bounded-time reachability problem for continuous-time linear dynamical systems can be encoded in $\reals_{\exp, \cos \restr [0, T]}$, we have the following (conditional) decidability result.
\begin{corollary}
	The continuous robust reachability problem for the continuous-time linear dynamical systems is decidable subject to Schanuel's conjecture.
\end{corollary}
\end{document}